\documentclass[journal]{IEEEtran}

\usepackage{cite}
\usepackage{array}
\usepackage{color}
\usepackage{float}
\usepackage[cmex10]{amsmath}
\usepackage{nomencl}						
\usepackage[normalem]{ulem}			        
\usepackage{diagbox}						
\usepackage{slashbox}						
\usepackage{colortbl}						
\usepackage{multirow}						
\usepackage{tabularx}
\usepackage{placeins}
\usepackage{algorithm}
\usepackage{mathrsfs}
\usepackage{mathdots}
\usepackage{amssymb}
\usepackage{amsthm}
\usepackage{arydshln}
\usepackage{color}
\usepackage[noend]{algpseudocode}
\usepackage{bm}
\usepackage{url}
\usepackage[hidelinks]{hyperref}
\usepackage[T1]{fontenc}
\newcommand{\subparagraph}{}
\usepackage{fancyhdr} 
\usepackage{extarrows}
\usepackage{wrapfig}
\usepackage{ragged2e}

\ifCLASSINFOpdf
  \usepackage[pdftex]{graphicx}
	\graphicspath{{./figure/}}
  \DeclareGraphicsExtensions{.pdf,.jpeg,.png}
\else
  \usepackage[dvips]{graphicx}
  \graphicspath{{./figure/}}
  \DeclareGraphicsExtensions{.eps}
\fi

\ifCLASSOPTIONcompsoc
  \usepackage[caption=false,font=normalsize,labelfont=sf,textfont=sf]{subfig}
\else
  \usepackage[caption=false,font=footnotesize]{subfig}
\fi

\newcommand{\figref}[1]{\figurename~\ref{#1}}

\newcommand{\diag}{\text{diag}}

\ifCLASSOPTIONonecolumn

\else

\fi

\newtheorem{lemma}{Lemma}

\begin{document}
\clearpage
\thispagestyle{plain}
\twocolumn[
\begin{@twocolumnfalse}
	\begin{center}
		\vspace{4cm}
		\LARGE
		\textbf{Copyright Statements}
		\vspace{1cm}
	\end{center}
	\justifying
	\Large
	This work has been submitted to the IEEE for possible publication. Copyright may be transferred without notice, after which this version may no longer be accessible.
	\justifying
\end{@twocolumnfalse}
]

\clearpage
\bstctlcite{IEEEexample:BSTcontrol}
\setcounter{page}{1}

\title{Participation Analysis in Impedance Models: The Grey-Box Approach for Power System Stability}
\author{Yue~Zhu, \IEEEmembership{Student Member, IEEE}, Yunjie~Gu, \IEEEmembership{Senior Member, IEEE}, Yitong~Li, \IEEEmembership{Member, IEEE}, Timothy~C.~Green, \IEEEmembership{Fellow, IEEE}
}

\IEEEaftertitletext{\vspace{-1.5\baselineskip}}

\ifCLASSOPTIONpeerreview
	\maketitle 
\else
	\maketitle
\fi

\thispagestyle{fancy}
\lhead{This work has been submitted to the IEEE for possible publication. Copyright may be transferred without notice, after which this version may no longer be accessible.}
\rhead{\thepage}
\cfoot{}
\renewcommand{\headrulewidth}{0pt}
\pagestyle{fancy}

\begin{abstract}
This paper develops a grey-box approach to small-signal stability analysis of complex power systems that facilitates root-cause tracing without requiring disclosure of the full details of the internal control structure of apparatus connected to the system. The grey-box enables participation analysis in impedance models, which is popular in power electronics and increasingly accepted in power systems for stability analysis. The Impedance participation factor is proposed and defined in terms of the residue of the whole-system admittance matrix. It is proved that, the so defined impedance participation factor equals the sensitivity of the whole-system eigenvalue with respect to apparatus impedance. The classic state participation factor is related to the impedance participation factor via a chain-rule. Based on the chain-rule, a three-layer grey-box approach, with three degrees of transparency, is proposed for root-cause tracing to different depths, i.e. apparatus, states, and parameters, according to the available information. The association of impedance participation factor with eigenvalue sensitivity points to the re-tuning that would stabilize the system. The impedance participation factor can be measured in the field or calculated from the black-box impedance spectra with little prior knowledge required.
\end{abstract}


\makenomenclature
\nomenclature[01]{$K, k$}{total nodes and node index, $k \in \left\{1,2,\cdots,K\right\}$}
\nomenclature[03]{$N$}{total state number}
\nomenclature[04]{$n,m$}{state and mode index, $n,m \in \left\{1,2,\cdots,N\right\}$}
\nomenclature[07]{$\hat{Y}$}{whole-system admittance matrix}
\nomenclature[08]{${Y}_\text{net}$}{nodal admittance matrix for the network}
\nomenclature[08]{${Z}$}{terminal impedance of grid-connected apparatus}
\nomenclature[09]{$x,z$}{state and mode vectors}
\nomenclature[10]{$A,a_{mn}$}{state matrix and its elements}
\nomenclature[11]{$\Lambda, \lambda$}{eigen matrix and eigenvalue}
\nomenclature[12]{$\Psi,\psi_{nm}$}{left-eigenvector matrix and its elements}
\nomenclature[13]{$\Phi,\phi_{mn}$}{right-eigenvector matrix and its elements}
\nomenclature[14]{$\rho$}{system parameter}
\nomenclature[16]{$\text{Res}_\lambda G$}{residue of $G$ at $\lambda$}
\nomenclature[17]{$p_{mn}$}{state participation factor}
\nomenclature[18]{$p_{\lambda,Z}$}{impedance participation factor}
\nomenclature[19]{$p_{\lambda,\rho}$}{parameter participation factor}
\nomenclature[20]{A}{apparatus in a power system}
\nomenclature[21]{$\langle \cdot,\cdot \rangle$}{Frobenius inner product}
\nomenclature[22]{{$\lVert\,{\cdot}\,\rVert$}}{Frobenius norm}
\nomenclature[23]{$^\top$}{matrix transpose}
\nomenclature[24]{$\overline{\phantom{x}}$}{complex conjugate}
\nomenclature[25]{$^*$}{conjugate transpose: $(\cdot)^* = \overline{({\cdot})}^\top$}

\ifCLASSOPTIONonecolumn
	\printnomenclature[2cm]
\else
	\printnomenclature[1.3 cm]
\fi

\section{Introduction}

Stability analysis for power systems must keep pace with the fast changing characteristics of the systems caused by emerging inverter-based resources (IBRs) replacing synchronous generators (SGs) and becoming the dominant sources. Unstable oscillations induced by IBRs are reported worldwide and the characteristics of such oscillations are distinctly different to the behaviour of a conventional SG-based grid \cite{bialek2020does,gu2019motion}. The understanding of the mechanisms of IBR induced instability is not yet comprehensive, nor are systematic solutions available for ensuring system-wide stability.   

One of the impediments to stability analysis is the lack of standardized and precise analytical models of the dynamic behaviour of IBRs. Unlike SGs, whose behaviour is largely determined by physics, IBR behaviour is shaped by its internal control system which is extremely flexible and far from being standardized as yet \cite{li2021impedance}. Manufacturers regard their control algorithms as critical proprietary technology and prefer to disclose to system operators only black-box models. 

The black-box models are usually given in the form of binary library files compiled from control algorithms, which are invoked by a numerical solver for time-domain simulation. The binary models can represent nonlinear dynamics with high fidelity, and therefore provides convincing results for stability validation. However, time-domain simulation is time-consuming and the results are not interpretive regarding the underlying mechanism. As an alternative to binary black-box models in time domain, attention recently falls upon impedance (or equivalently, admittance) models in frequency domain \cite{wang2017unified}. Impedance models represent linearized small-signal dynamics via the frequency spectra of input-output relationships and therefore are also black-box models. Impedance models offer some extent of interpretability in the form of resonance peaks and stability margins, but such interpretability is limited to simple systems with explicit bipartition, e.g. a single IBR connected to an infinite bus \cite{gu2019motion}.

Due to the lack of interpretability of black-box models, system operators prefer white-box models for stability analysis. White-box models are state-space equations containing all of the physical and control states of each apparatus. The eigenvectors of state-space matrix yield the participation factor indicating correlations between states and modes. The participation factor can be used to trace the root-cause of unstable or under-damped oscillations and thus provide very high interpretability \cite{rommes2008computing,Sinha2019}.

To bridge the gap between the black-box models divulged by manufacturers and the white-box models desired by operators, efforts have been made to implement participation analysis in impedance models. A whole-system impedance model is presented in \cite{gu2020impedance}, which formulates the dynamics of the grid as the whole-system admittance or impedance matrix. All of the elements in the admittance and impedance matrices share the same poles (equivalent to the eigenvalues of the system) and therefore contain the same information regarding system stability, but different elements may have different resonant peaks at each of the poles which reflect their differing participation in the corresponding mode. Further to the whole-system model, \cite{Huang2007,ebrahimzadeh2018bus,zhan2019frequency} introduce nodal participation factor and branch participation factor, based on the eigenvectors of the nodal admittance matrix and loop impedance matrix. These two approaches open up paths towards participation analysis with impedance model, but the relationship between the participation factor in impedance model and that in state-space model is still unclear. As a result, the participation factor in impedance model, defined in this way, cannot be used to look inside a black-box and so cannot determine exactly which state causes an instability, nor indicate how internal parameters should be re-tuned to stabilize the system. 

In this paper, we first prove that the sensitivity of state-space eigenvalue with respect to apparatus impedance equals the residue of the whole-system admittance seen at this apparatus. Here we use apparatus to refer to any device connected in shunt at a node, such as SG, synchronous condenser, IBR, active load, or FACTS device. Based on this finding, we define the residue as impedance participation factor in a way that is consistent with the classic state participation factor. We further prove the chain-rule for the impedance participation factor, which links the state and the parameter participation factors. This key step allows one to look inside black-box models without disclosing the internal details. To highlight this feature, we name our method the \textit{grey-box approach}. The proposed grey-box approach has not only a rigorous mathematical basis but also good potential for practical applications. The residues can be measured in the field or estimated from black-box impedance spectra with very little prior knowledge and few presumptions.

The paper is organized as follows. The whole-system admittance model is briefly introduced in Section II. The theory of the impedance participation factor and the corresponding grey-box approach is presented in Section III. The major findings are verified and applications are demonstrated on a NETS-NYPS 68-bus system in Section IV. The last section concludes the paper.

\section{Whole-System Admittance Model} \label{section2}
This section briefly describes the whole-system admittance model proposed in \cite{gu2020impedance} based on which participation analysis will be conducted in the next section.  
In a \textit{K}-node meshed network, as illustrated in \figref{fig1_wholesys}, virtual voltage injections $\hat{v} = [\hat{v}_1, \hat{v}_2, \cdots, \hat{v}_K]^\top$ are introduced in series with each item of apparatus in the system, such as a SG or an IBR, to create perturbations in the corresponding currents $\hat{i} = [\hat{i}_1, \hat{i}_2, \cdots, \hat{i}_K]^\top$. The transfer function matrix from $\hat{v}$ to $\hat{i}$ is called the whole-system admittance $\hat{Y}$, that is, ${\hat{i}(s)}={\hat{Y}}(s) \cdot {\hat{v}}(s)$. Following \cite{gu2020impedance}, $\hat{Y}$ is given by
\begin{equation}
\label{Y_whole_system}
{\hat{Y}}=\left( {I}+{Y}_{\text{net}}{Z} \right) ^{-1}{Y}_{\text{net}}
\end{equation}
where
\begin{equation}
{Z}=\text{diag}\left( Z_{1},Z_{2} ,\cdots ,Z_{K} \right)
\end{equation}
is a diagonal matrix containing the terminal impedance of all apparatus, and $Y_{\text{net}}$ is the nodal admittance matrix of the network.

$\hat{Y}$ is a $K \times K$ transfer function matrix. All elements of $\hat{Y}$ share the same poles which are equivalent to the eigenvalues of the system \cite{gu2020impedance}. On the other hand, different elements may have different levels of excitation and thus different magnitude of resonant peaks at each pole of the system which reflect the different participation levels in each mode. This characteristic serves as the basis to extract precise participation factors from $\hat{Y}$ matrix which is the central theme of this paper.

The diagonal elements of $\hat{Y}$ have a clear physical meaning. Taking the first diagonal element $\hat{Y}_{11}$ as an example, as illustrated in \figref{fig1_wholesys}, $\hat{Y}_{11}$ is essentially the admittance for the loop containing ${Z}_{1}$ and ${Z}_{\text{g}1}$, in which ${Z}_{1}$ is the terminal impedance of the apparatus at the first node and ${Z}_{\text{g}1}$ is the impedance for the rest of the grid as seen by ${Z}_{1}$ at that node. Therefore, we have $\hat{Y}_{11} = ({Z}_{1} + {Z}_{\text{g}1})^{-1}$ and the same principle holds for any node $k$
\begin{equation}
\label{loop_y_z}
    \hat{Y}_{kk} = ({Z}_{k} + {Z}_{\text{g}k})^{-1}.
\end{equation}
As we shall see in the next section, this property is very useful for impedance-based participation analysis.

\begin{figure}
\centering
{\includegraphics[scale = 0.8]{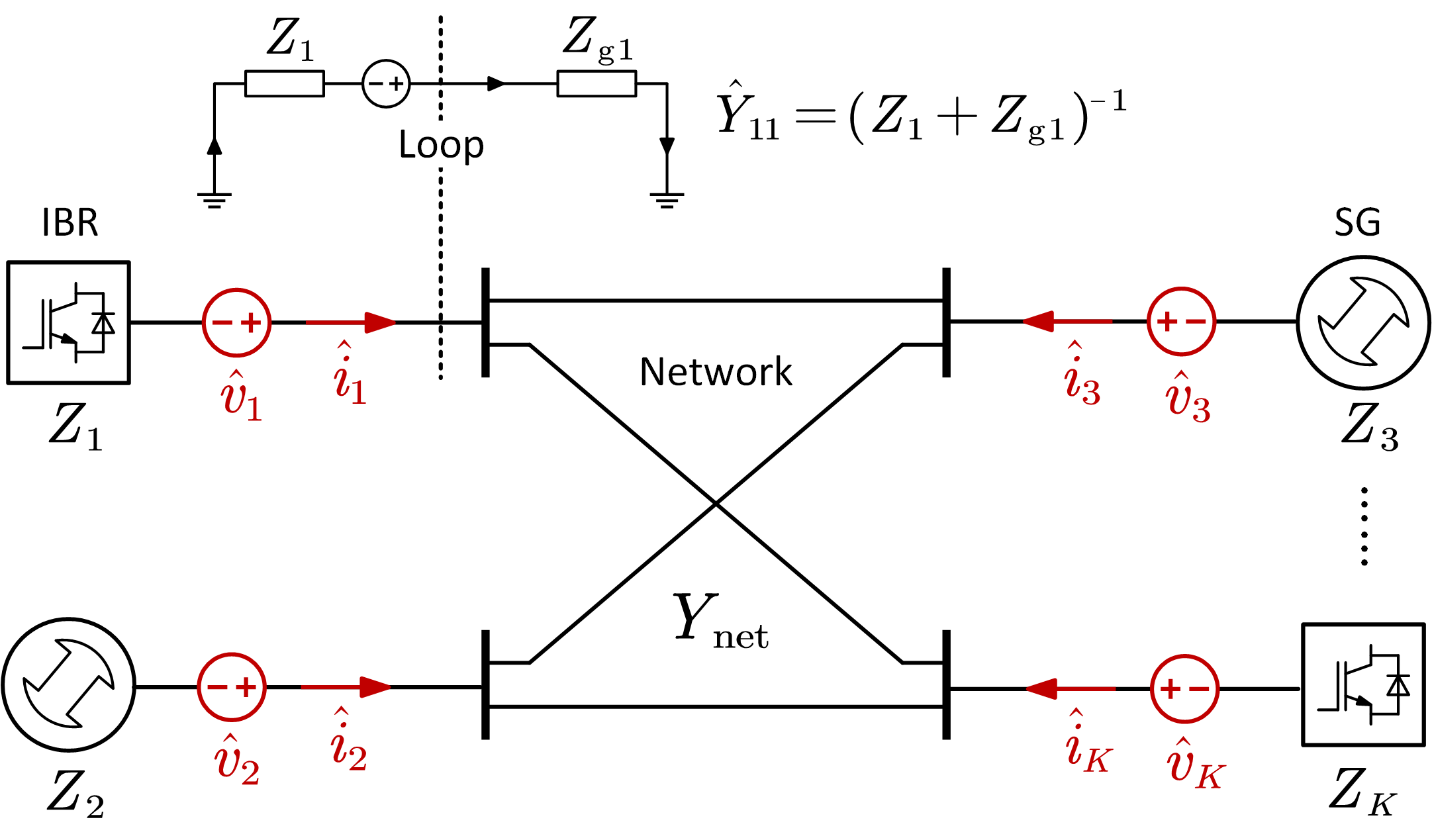}}
\caption{Illustration of the whole-system admittance model. Virtual voltage injection $\hat{v}$ to excite current perturbation $\hat{i}$.}
\label{fig1_wholesys}
\end{figure}


\section{Impedance Participation Factor and the Grey-Box Approach}

Before proceeding to impedance-based participation analysis, we review the classic state-based participation analysis so that the relationship between them is revealed to readers.

\subsection{State Participation Factor}
The dynamics of a power system linearized around its equilibrium point can be represented by a state equation
\begin{equation}
\label{state_equation}
\dot{x} = A x.
\end{equation}
This state equation is usually very high-order but can be decomposed into a series of first-order equivalents via coordinate transformations $z = \Psi x$ and $x = \Phi z$ ($\Phi = \Psi^{-1}$) such that the state matrix in the new coordinate $z$ is diagonalized \cite{kundur1994power}, that is,
\begin{equation}
\label{z_equation}
\dot{z} = \Lambda z, \  \Lambda = \Psi A \Phi = \diag{(\lambda_1,\lambda_2,\cdots,\lambda_N)}
\end{equation}
where $\lambda_n$ ($n=1,2,\cdots,N$) is the $n$-th eigenvalue of $A$, and the rows and columns of $\Psi$ and $\Phi$ correspond to the left- and right-eigenvectors of $A$, respectively. $z$ is a vector of modes of the system and determines stability according to the corresponding eigenvalues. The transformation matrices $\Psi$ and $\Phi$ describe the back-and-forth correlations between modes $z$ and states $x$. Putting $\Psi$ and $\Phi$ together, the participation factor of the $m$-th state $x_m$ in the $n$-th mode $z_n$ is defined as \cite{Perezarriaga1982}
\begin{equation}
\label{state_pf}
    p_{mn} = \psi_{nm} \phi_{mn},
\end{equation}
 where $\psi_{nm}$ and $\phi_{mn}$ are elements of $\Psi$ and $\Phi$, respectively. In this paper, $p_{mn}$ is renamed \textit{state participation factor} to distinguish it from the \textit{impedance participation factor} to be introduced in the next subsection.

The state participation factor in (\ref{state_pf}) is defined heuristically but is endowed with a rigorous mathematical meaning due to its linkage to eigenvalue sensitivity \cite{kundur1994power}, that is
\begin{equation}
\label{eigen_sensitivity}
    p_{mn} = \frac{\partial \lambda_n}{\partial a_{mm}} \Rightarrow \Delta \lambda_n = p_{mn} \Delta a_{mm}. 
\end{equation}
In (\ref{eigen_sensitivity}), $a_{mm}$ is the $m$-th diagonal element of the state matrix $A$ and can be interpreted as the local dynamics for $x_m$ itself as if decoupled from the other states. In contrast, $\lambda_n$ describes the global dynamics for the whole system. Therefore, (\ref{eigen_sensitivity}) implies that $p_{mn}$ determines how the local dynamics affect the global dynamics. In the light of this interpretation, state participation factor can not only identify which local states are participating in a particular whole-system mode, but can also indicate how local parameters should be re-tuned to better damp that mode. For this reason, the state participation factor has become an important tool in stability analysis enabling tracing of root-causes, and trouble-shooting of complex power systems \cite{kundur1994power,Perezarriaga1982,rommes2008computing}.

We now add inputs and outputs to the state equation (\ref{state_equation}) to see how the state participation factor is reflected in the transfer functions of the system. For a system with input $u$ and output $y$ 
\begin{equation}
\label{state_equation_io}
\begin{array}{l}
\dot{x} = A x + Bu \\
y = C x
\end{array}
\end{equation}
the transfer function from $u(s)$ to $y(s)$ is 
\begin{equation}
\label{state_to_tf}
G(s) = C(sI-A)^{-1}B =  C \Phi (sI-\Lambda)^{-1} \Psi B.
\end{equation}
For the sake of brevity, we take a case of a  single-input and single-output (SISO) system, i.e., $u(s)$ and $y(s)$ are scalars. In such a case, $G(s)$ can be expanded as
\begin{equation}
\label{eq_tf}
G(s) = \sum_{n=1}^{N} \frac{\text{Res}_{\lambda_n}}{s-\lambda_n} 
\end{equation}
where 
\begin{equation}
\mathrm{Res}_{\lambda _n}=\sum_{m=1}^N{c_m}\phi _{mn}\sum_{m=1}^N{\psi _{nm}b_m} 
\end{equation}
is the residue of $G(s)$ at $\lambda_n$, and $b_m$ and $c_m$ are elements of $B$ and $C$, respectively. It is clear that the eigenvalues of the state matrix appear as poles in the transfer function. If $C$ and $B$ are in such a form that their $m$-th elements equal 1 and all other elements equal 0, $\text{Res}_{\lambda_n}$ can be simplified to 
\begin{equation}
\label{hint}
\text{Res}_{\lambda_n} = \psi_{nm} \phi_{mn}  = p_{mn}.
\end{equation}
In this special case, the state participation factor is the same as the residue of the transfer function, an observation which hints at residues being useful for participation analysis. However, (\ref{hint}) is based on a strong assumption about $B$ and $C$ which may not hold for common cases. In the next subsection, we explore the general residue-participation relationship which yields the impedance participation factor.

\subsection{Impedance Participation Factor}
In order to clarify the roles of residues in participation analysis in a general form, we introduce Lemma \ref{Lemma}.

\begin{lemma}
\label{Lemma}
For a square transfer function matrix $G_\rho$ depending on parameters $\rho$, let $H_\rho$ be the inverse transfer function of $G_\rho$, i.e. $H_\rho = G_\rho^{-1}$, and $\lambda$ be a non-repeated pole of $G_\rho$. When the parameters $\rho$ are perturbed infinitesimally by $\Delta \rho$, $\lambda$ and $H_\rho$ are perturbed by $\Delta \lambda$ and $\Delta H_\rho$ correspondingly and we have the following relationship in between:
\begin{equation} \label{eq_lemma}
\Delta \lambda = \langle - \mathrm{Res}^{*}_\lambda G_\rho, \Delta H_\rho(\lambda) \rangle
\end{equation}
in which $\langle \cdot, \cdot\rangle$ is the Frobenius inner product of two matrices, $\mathrm{Res}_\lambda G_\rho$ is the residue matrix of $G_\rho$ at $\lambda$ (the residue operates element-wise on a matrix), $^*$ denotes the conjugate transpose of the residue matrix, and the equation holds in the sense of neglecting high-order infinitesimals. 
\end{lemma}

\begin{proof}
The proof of the Lemma is given in Appendix B, and a brief introduction to the mathematical preliminaries used in this Lemma is included in Appendix A. A simple example to illustrate this Lemma is provided in Appendix C.
\end{proof}

If we take $G_\rho$ to be the whole-system admittance seen at node $k$ defined in Section II, that is, $G_\rho = \hat{Y}_{kk}$, the corresponding $H_\rho$ is
\begin{equation}
H_\rho = \hat{Y}_{kk}^{-1} = Z_{k} + Z_{\text{g}k}
\end{equation}
according to (\ref{loop_y_z}). When $Z_{k}$ itself is subject to a perturbation we can say that $Z_{\text{g}k}$, representing the rest of the grid, is unchanged, that is, $\Delta Z_{\text{g}k} = 0$, so we have
\begin{equation}
\Delta H_\rho = \Delta Z_{k} + \Delta Z_{\text{g}k} = \Delta Z_{k}
\end{equation}
and hence
\begin{equation} 
\label{sense_z}
\Delta \lambda = \langle -\text{Res}^{*}_\lambda \hat{Y}_{kk}, \Delta Z_{k}(\lambda) \rangle.
\end{equation}
As explained in the previous subsection, the poles of the transfer function $\hat{Y}_{kk}$ are exactly the eigenvalues of the system. Therefore, (\ref{sense_z}) implies that the sensitivity of an eigenvalue to an apparatus impedance is determined by the residue of the whole-system admittance seen by the same apparatus. Since sensitivity is equivalent to participation, we define the residue as the \textit{impedance participation factor}
\begin{equation} 
\label{impedance_participation}
p_{\lambda,Z_{k}} \triangleq -\text{Res}^{*}_\lambda \hat{Y}_{kk}
\end{equation}
such that
\begin{equation} 
\label{pf_z_lambda}
\Delta \lambda = \langle p_{\lambda,Z_{k}}, \Delta Z_{k}(\lambda) \rangle.
\end{equation}

Due to impedance-admittance duality, we can similarly define the \textit{admittance participation factor} as the residue of the whole-system impedance
\begin{equation} 
p_{\lambda,Y_{k}} \triangleq -\text{Res}^{*}_\lambda \hat{Z}_{kk}
\end{equation}
such that
\begin{equation} 
\Delta \lambda = \langle p_{\lambda,Y_{k}}, \Delta Y_{k}(\lambda) \rangle
\end{equation}
where $Y_{k}$ is the admittance of the $k$-th apparatus and $\hat{Z}_{kk}$ is the whole-system impedance \cite{gu2020impedance} at the $k$-th node. The impedance and admittance participation factors are equivalent theoretically but each may better serve different applications where impedance or admittance is more readily available. We focus on the impedance participation factor in this paper.

If we know the sensitivity of the apparatus impedance against its parameters $\rho$, that is, 
\begin{equation} 
\label{parameter_z_sens}
\Delta Z_{k}(\lambda) = \frac{\partial Z_{k}(\lambda)}{\partial \rho} \cdot \Delta \rho
\end{equation} 
we further define the \textit{parameter participation factor}
\begin{equation} 
\label{parameter_pf}
p_{\lambda,\rho} = \left< p_{\lambda, Z_{k}}, 
\frac{\partial Z_{k}(\lambda)}{\partial \rho} \right> 
\end{equation}
such that 
\begin{equation} 
\label{chain_rule}
\Delta \lambda = p_{\lambda,\rho} \Delta \rho.
\end{equation}
If $\rho$ is selected as $a_{mm}$ in the state matrix $A$, the corresponding parameter participation factor is the state participation factor 
\begin{equation}
\label{pf_relationships}
p_{mn} =
p_{\lambda_n,a_{mm}} = \left< p_{\lambda_n, Z_{k}},
\frac{\partial Z_{k}(\lambda_n)}{\partial a_{mm}} \right>. 
\end{equation}
Thus we establish the relationship between the different types of participation factor and summarize this relationship in \figref{fig_pf_summary}. The state coefficient $a_{mm}$ might be the combination of multiple physical or control parameters and thus is not an independent parameter itself. We treat $a_{mm}$ as a virtual parameter to draw the linkage between state and parameter participation factors, as marked by the dashed arrow in \figref{fig_pf_summary}.

\begin{figure}
\centering
\includegraphics[scale=0.72]{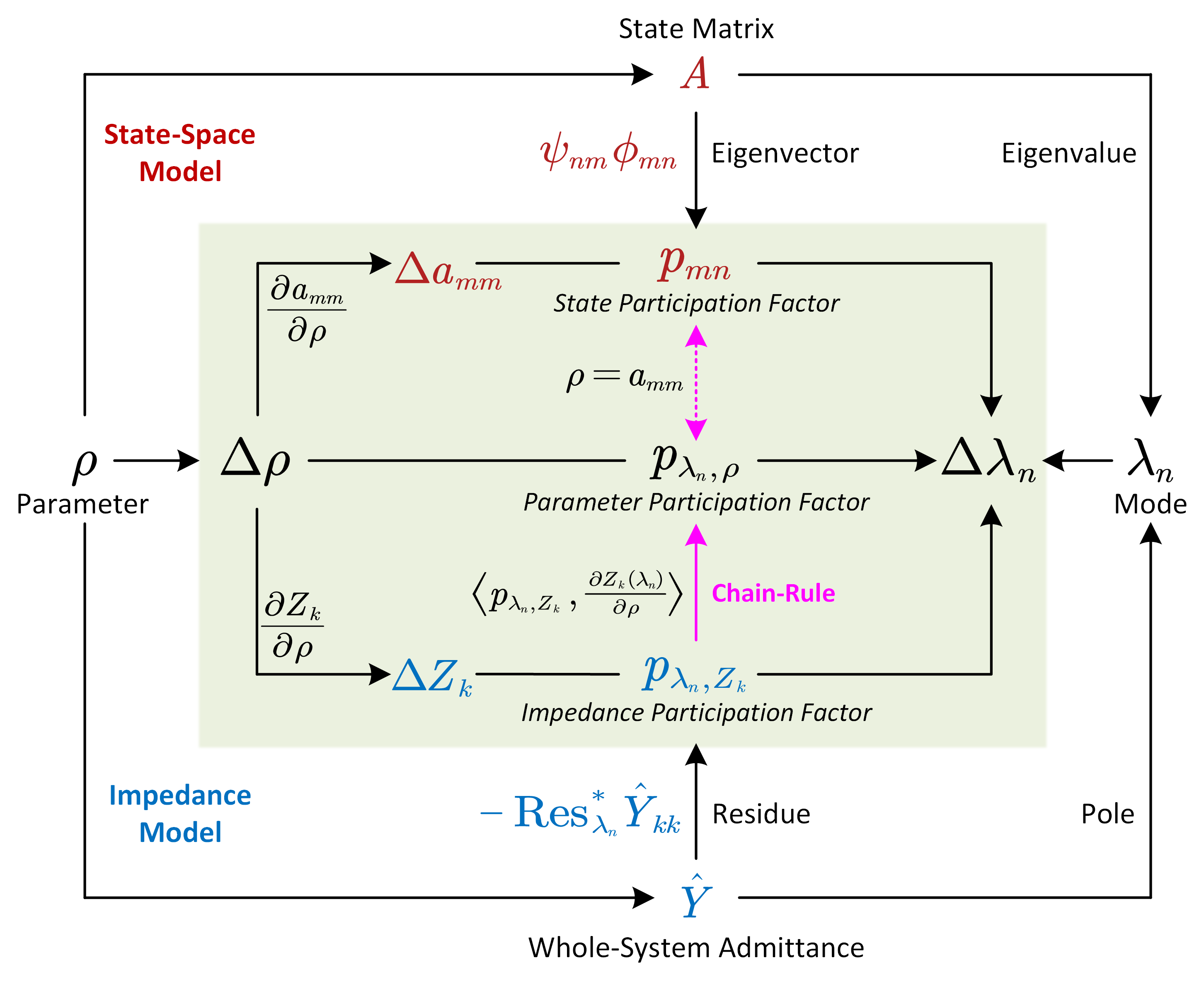}
\caption{The relationship between impedance participation factor and state participation factor and the chain-rule for participation propagation.}
\label{fig_pf_summary}
\end{figure}

Equations (\ref{parameter_pf}) and (\ref{pf_relationships}) are called the chain-rule of participation factors which is of profound importance in participation analysis. The impedance participation factor enables us to evaluate the participation of an apparatus in system oscillations through only black-box models. On top of this, the chain-rule yields the state and parameter participation factors, further enabling us to look inside the black-box and trace root-causes to detailed parameters and states without disclosing the state equation. Based on this chain-rule of participation factors, we proposed the grey-box approach for power system stability analysis, as will be described in the following subsection.

\subsection{The Grey-Box Approach}
The grey-box approach contains three layers with different transparencies according to the available prior knowledge, as illustrated in \figref{fig_greybox}. The higher the transparency, the more prior knowledge is needed but along with that comes more useful information for root-cause tracing and trouble-shooting in whole-system stability analysis. Now we describe in detail each layer and the relationships between them.

\begin{figure}
\centering
\includegraphics[scale=0.75]{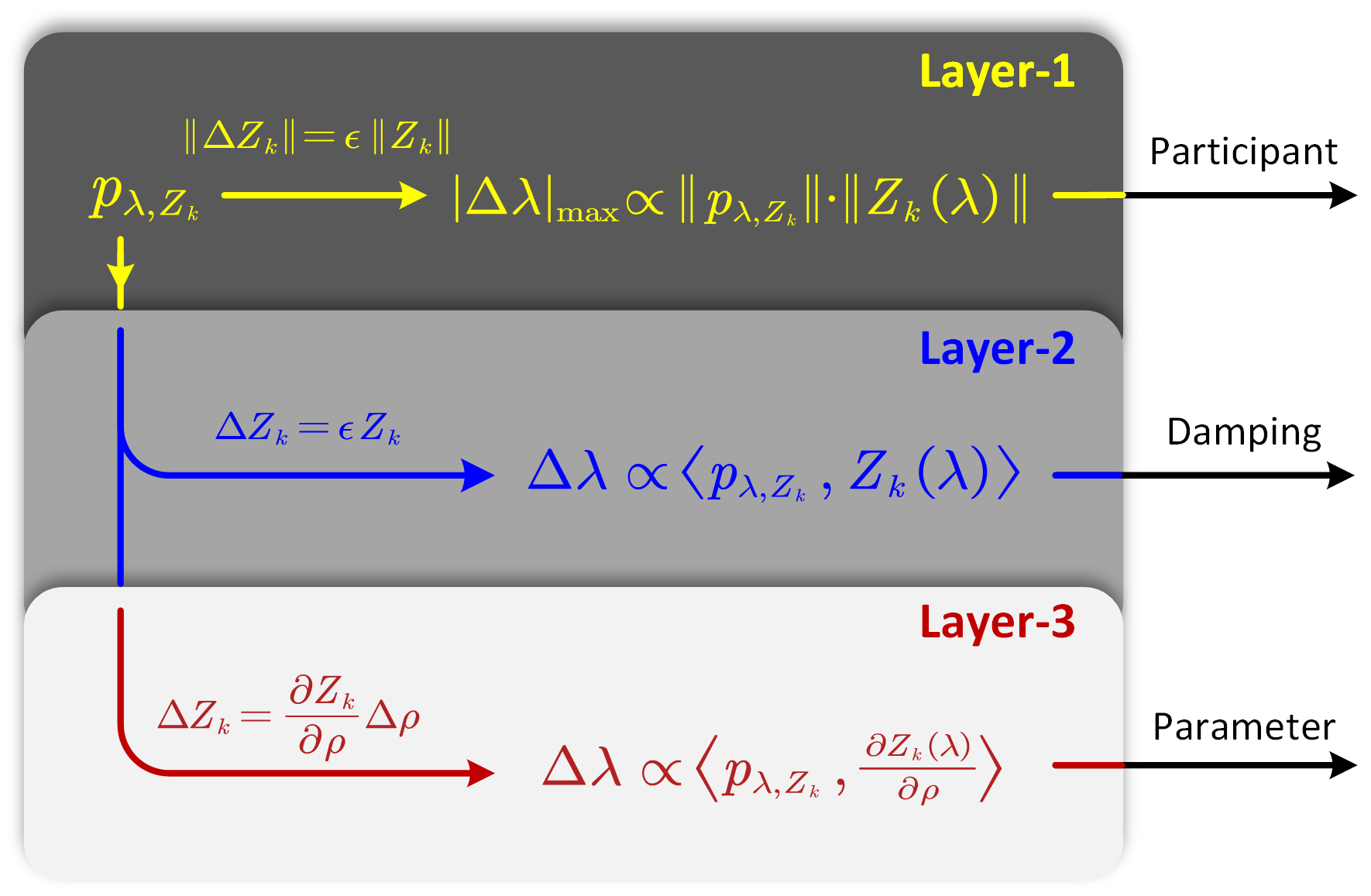}
\caption{Illustration of the three-layer grey-box. In Layer-1, estimates of the potential participants are created based on the upper bound of $\Delta \lambda$ subject to $\| \Delta Z_k \|=\epsilon \| Z_k \|$. In Layer-2, the contribution of a participant to mode damping is estimated based on the real part of $\Delta \lambda$ subject to $ \Delta Z_k =\epsilon  Z_k $. In Layer-3, the root-cause of instability within the participating apparatus is identified, and parameter re-tuning facilitated, via the impedance-parameter sensitivity.}
\label{fig_greybox}
\end{figure}

\subsubsection{Grey-Box Layer-1} 
In this first layer, the only prior knowledge available is the impedance $Z_k$ for all apparatus in the system, along with the nodal admittance matrix $Y_\text{net}$ determined by the topology and line impedance of the network. The impedance $Z_k$ can be provided either in the form of transfer functions $Z_k(s)$ or frequency spectra $Z_k(j\omega)$. For transfer functions, the impedance participation factor $p_{\lambda_n,Z_k}$ can be calculated directly, but for frequency spectra, the impedance participation factor needs to be estimated indirectly.

Due to the three-wire, three-phase nature of power systems, $Z_k$ and $p_{\lambda_n,Z_k}$ are $2 \times 2$ matrix blocks in the synchronous $dq$ frame. We need a scalar index to represent the four elements of $p_{\lambda_n,Z_k}$ collectively so that different $p_{\lambda_n,Z_k}$ can be compared and the location of the dominant apparatus (for a given mode) determined. To this end, we assign a consistent magnitude perturbation to each $Z_k$ and observe the effect on the eigenvalue $\lambda$. The perturbation is normalized to $\| Z_k \|$ so that it scales with the corresponding apparatus, that is, $\| \Delta Z_k \| = \epsilon \| Z_k \|$, where $\epsilon$ is a small positive constant. According to the Cauchy inequality, we have
\begin{equation}
\label{cauchy}
|\Delta \lambda| = |\langle p_{\lambda,Z_{k}}, \Delta Z_{k}(\lambda) \rangle |
 \leq \|  p_{\lambda,Z_{k}} \| \cdot \|  \Delta Z_{k}(\lambda) \|
\end{equation}
which yields
\begin{equation}
\label{lambda_max}
|\Delta \lambda|_\text{max} = 
\|  p_{\lambda,Z_{k}} \| \cdot \|  \Delta Z_{k}(\lambda) \| =
\epsilon \|  p_{\lambda,Z_{k}} \| \cdot \|  Z_{k}(\lambda) \|.
\end{equation}
It is clear from (\ref{lambda_max}) that $\| p_{\lambda,Z_{k}} \| \cdot \|Z_k(\lambda)\|$ determines the upper bound of 
$\Delta \lambda$, $|\Delta \lambda|_\text{max}$, 
which is the maximum possible participation of the corresponding apparatus. Only apparatus with relatively large $| \Delta \lambda |_ \text{max}$ may possibly, but not necessarily, participate in the $\lambda$-mode. Thus, we use $\| p_{\lambda,Z_{k}} \| \cdot \|Z_k(\lambda)\|$ as the primary participation index in Layer-1 of the grey-box approach. Layer-1 roughly identifies potential participants in a mode and does so with very little prior knowledge. 

\subsubsection{Grey-Box Layer-2}
Building on Layer-1, Layer-2 adds a stipulation that the perturbation of apparatus impedance $\Delta Z_k$ is aligned to the original impedance $Z_k$, that is, $\Delta Z_k = \epsilon \, Z_k$, 
where $\epsilon$ is a very small positive real number. This is a reasonable stipulation because it emulates the effect of scaling up or down an apparatus so the resulted impedance shrinks or grows in amplitude but maintains the same angle. The scaling can be done by changing the base power of an apparatus with per-unit parameters, or by connecting or disconnecting a portion of apparatus from a farm of identical apparatus (e.g. a wind farm). Layer-2 does not require extra information to Layer-1 but adds to Layer-1 by the extra stipulation on the orientation of $\Delta Z_k$, which yields the direction of $\Delta \lambda$ and thus brings additional knowledge about the node's impact on system damping:
\begin{equation}
\label{layer_2}
\begin{split}
\Delta \lambda &= \langle p_{\lambda,Z_{k}}, \Delta Z_{k}(\lambda) \rangle = \epsilon \left< p_{\lambda,Z_{k}}, {Z_{k}(\lambda)} \right>.
\end{split} 
\end{equation}
From (\ref{layer_2}) we see that $\Delta \lambda$ is determined by $\left< p_{\lambda,Z_{k}}, {Z_{k}(\lambda)} \right>$ which we use as the new participation index for Layer-2. If the real-part of $\left< p_{\lambda,Z_{k}}, {Z_{k}(\lambda)} \right>$ is positive, it implies that scaling up the corresponding apparatus connected at the node tends to stabilize the system, and vice versa.

\subsubsection{Grey-Box Layer-3} 
The final layer of the grey-box approach aims to look into the participating apparatus in order to identify which physical component or control loop in the apparatus is the root-cause of instability, and thus provide information on which parameter should be re-tuned, and how to re-tune it, to stabilize the system. 

Layer-3 uses the sensitivity of an impedance to its internal parameters, ${\partial Z_{k}}/{\partial \rho}$, so that a parameter perturbation is propagated to an impedance perturbation and further into an eigenvalue perturbation via the chain-rule (see (\ref{parameter_z_sens})-(\ref{chain_rule}) in Section III-B). The chain-rule yields the parameter participation factors that indicates which internal parameters can be re-tuned so that the particular eigenvalue will move into the desired direction on complex plane. Importantly, the impedance-parameter sensitivity ${\partial Z_{k}}/{\partial \rho}$ discloses little information regarding the internal design or control of an apparatus and yet it enables root-cause tracing inside the apparatus as effectively as a transparent white-box model. This is a great advantage of the Layer-3 grey-box approach. 

\subsection{Practical Implementation}

Now we discuss the detailed implementation of the grey-box approach in practical applications. Four sets of data are required before applying the approach, namely, $Z_k$, $\hat{Y}_{kk}$, $p_{\lambda,Z_k}$, and, for Layer-3 only, ${\partial Z_{k}}/{\partial \rho}$. 
All data are presented as frequency spectra, that is, numerical values for a range of frequencies, so as to avoid symbolic calculation and ensure scalability to large-scale systems \cite{Yang2019Automation}.

There are two routes to obtain $Z_k$. 
\subsubsection{Model-based Route}   
This route is relevant to manufacturers who have available detailed analytical models that preserve every state. 
Dynamic differential equations can be derived from such models and can be linearized around an equilibrium point to obtain state-space matrices \cite{LinearizationMatlab}. The state-space matrices are then transformed into transfer functions using (\ref{eq_tf}) which yield frequency spectra by letting $s=j\omega$ in the transfer functions. This route is used for the case studies in this paper and readers may refer to \cite{Simplex} for detailed codes.

\begin{figure*}
	\centering
	\includegraphics[width=5in]{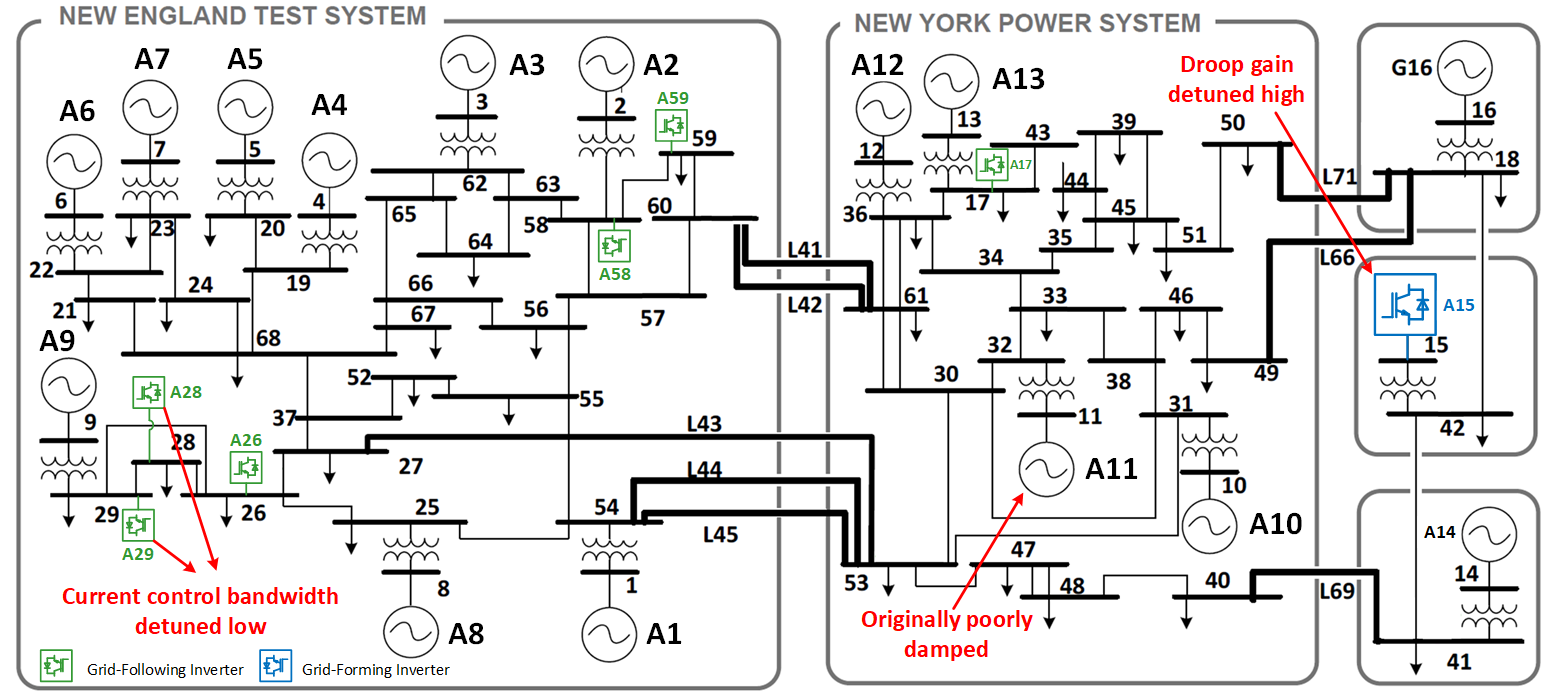}
	\caption{Modified NETS-NYPS 68 bus system, with 6 extra grid-following inverters, and a grid-forming inverter replacing the synchronous machine at bus 15.}
	\label{NETSNYPS68}
\end{figure*}

\subsubsection{Measurement-Based Route}

This route is applicable to both manufacturers and operators. An extra perturbation signal is applied to the inputs of the apparatus under test and the resulting perturbations in outputs are measured. The ratios of outputs to inputs in frequency domain are the spectra of the impedance (or admittance, depending on how the inputs and outputs are selected). The perturbation signal can either be a pseudorandom signal \cite{Roinila2013Broadband} or a swept-frequency signal \cite{Familiant2009New} to ensure sufficient excitation in all frequencies within the range. The measurement can be conducted in either hardware experiments or in electromagnetic transient (EMT) simulation. This route avoids hand-written differential equations, but the downside is that the measurement can be time-consuming and sensitive to measurement noise.
	
The impedance $Z_k$ is dependent upon parameter $\rho$, from which we calculate ${\partial Z_{k}}/{\partial \rho}$ numerically by
\begin{equation} 
\frac{\partial Z_k }{\partial \rho} \approx \frac{Z_{k, \Delta \rho} - Z_{k} }{\Delta \rho}
\end{equation}
where $\Delta \rho$ is a small perturbation and $Z_{k, \Delta \rho}$ is the impedance under the effect of the perturbation. $\Delta \rho$ is selected as $\Delta \rho = 10^{-5}(1+|\rho|)$ following \cite{LinearizationMatlab} as a trade-off between the relative and absolute definitions of perturbation. $Z_{k, \Delta \rho}$ must take into account the changes in the equilibrium point caused by the parameter perturbation which in turn affects impedance during linearization.  

Linking all $Z_k$ in the network according to (\ref{Y_whole_system}) yields the whole-system admittance $\hat Y_{kk}$. The poles and residues of $\hat Y_{kk}$ can be identified from its spectra by rational approximation \cite{Gustavsen1999}. The most time-consuming calculations in the grey-box approach are numerical matrix multiplications and inversions which have computational complexity of $O(N^3)$ or less using commonly available algorithms \cite{boyd2018introduction}. This is comparable to the computational complexity of the Newton-Raphson method commonly used in static power flow analysis. Therefore, the scaling of the grey-box approach to large systems is likely to be acceptable.
	
It is worth noting that participation factors are essentially local sensitivities and therefore concern at small parameter perturbations. If the parameters need to be tuned over a large range, the grey-box approach needs to be applied iteratively over the path of parameter variations. Such iteration is applicable to continuous parameters but not to discrete and logic parameters. In a case where two poles are very close to each other or poles are repeated, the corresponding residues can not be distinguished so the grey-box approach cannot be applied directly. This issue can also be solved by iteration, that is, treating the two non-distinguishable poles as the same pole and perturb parameters until they separate.

The grey-box approach is intended for small-signal analysis to be complemented with EMT simulation for large-signal analysis. The grey-box approach is interpretive regarding the underlying structure of system dynamics, whereas EMT simulation retains all non-linearity and thus provides high fidelity validation over various transients.

\color{black}

\section{Case Study of a Composite Power System}

We now demonstrate the grey-box approach through a case study of a composite power system including IBRs. The chosen system, \figref{NETSNYPS68}, is based on the NETS-NYPS 68 bus system \cite{Pal_Report2013, Qi_Identification} with six additional IBRs (Type-IV wind farms) connected to buses 17, 26, 28, 29, 58, 59. The SG at bus 15 is replaced with a grid-forming inverter. Each IBR is an aggregate  with a scaled rating representing many individual IBR. To make the system prone to oscillation, the frequency droop gain of A15 is deliberately de-tuned high, and the current control bandwidths of A28 and A29 are de-tuned by $-40\%$ and $-56\%$ respectively. All SGs use the same parameters as \cite{Pal_Report2013}, meaning that A11 is poorly damped and the least stable generator in the system\cite{Qi_Identification}. The system data, the codes used to generate the simulation results, and all numerical results can be found at: \url{https://github.com/Future-Power-Networks/Publications}\cite{Simplex}.

The whole-system admittance of the network, constructed from the impedance of all apparatus and admittance of all of the network lines, is displayed in the bode plot in \figref{fig_admitt}. At each node, the whole-system admittance $\hat{Y}_{kk}$ is a $2 \times 2 $ matrix in the synchronous $dq$ frame, but only one of the four elements in the matrix is displayed since that is suficient to illustrate the characteristics of the system. Only the nodes with sources (SGs or IBRs) present are plotted because the other nodes are passive. Several resonant peaks appear in the bode plot, each representing an oscillation mode in the system. The mode at 60~Hz arises from the flux dynamics of windings and lines, and is a standard feature \cite{kundur1994power}. The modes around 1-2~Hz are rotor swing modes of the SGs, and the mode at high frequency is caused by the $LCL$ filter of the grid-forming inverter. 
Three modes, annotated 1 to 3 in \figref{fig_admitt}, are selected for further analysis, and we use the grey-box approach to trace the root-cause of these modes and find ways to damp the modes.

\begin{figure*}
	\centering
	\includegraphics[width=5in]{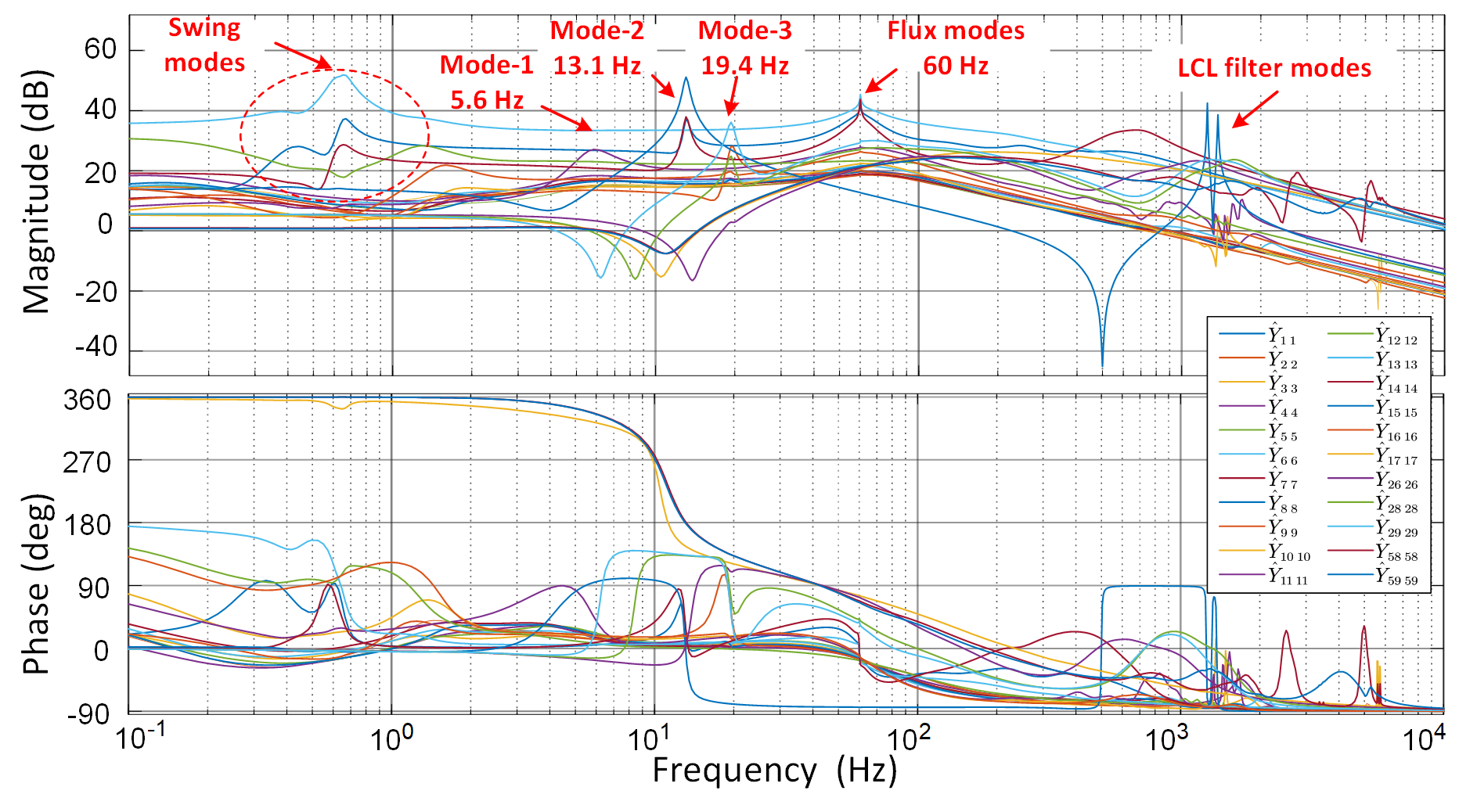}
	\caption{Bode diagram of whole system admittance $\hat{Y}_{kk}$ at nodes with sources, presented in $dd$ axis.}
	\label{fig_admitt}
	\vspace{-2mm}
\end{figure*}

\figref{fig_level2} shows the results of applying grey-box Layer-1 (on the left) and Layer-2 (on the right). For mode-1 (5.6~Hz), A11 stands out in the Layer-1 pie chart. Further, the breakdown into real and imaginary components in Layer-2 shows that A11 affects both the damping and natural frequency of mode-1 while the adjacent apparatus A10 and A12 also influence the damping of this mode. The negative real-part in Layer-2 indicates that scaling up the power rating of A15 (which is equivalent to decreasing its impedance) tends to destabilize the system. Similar analysis of mode-2 shows that A15 is dominant in this mode, affecting both the damping and natural frequency. Mode-3 is more complicated: Layer-1 reveals that there are multiple participants (A9, A28, and A29), and the Layer-2 decomposition shows that the SG (A9) and the IBR (A28, A29) have opposite signs for the component of $\Delta \lambda$. For A9, scaling up of power rating would improve stability whereas for A28 and A29 scaling up of the power rating decreases stability and indicates that mode 3 is an IBR-induced oscillation. Further, comparing A28 and A29, we see that A29 has a larger participation in this mode, which is attributable to the fact that A29 was de-tuned further than A28. The exact cause of the destabilization is not revealed until Layer-3 of the grey-box which can point to particular components and control parameters. Nonetheless, Layer-1 and Layer-2 grey-boxes reveal rich information about the root-causes of modes 1, 2 and 3 without significant prior knowledge.

\begin{figure}
	\centering
	\includegraphics[width=3.3in]{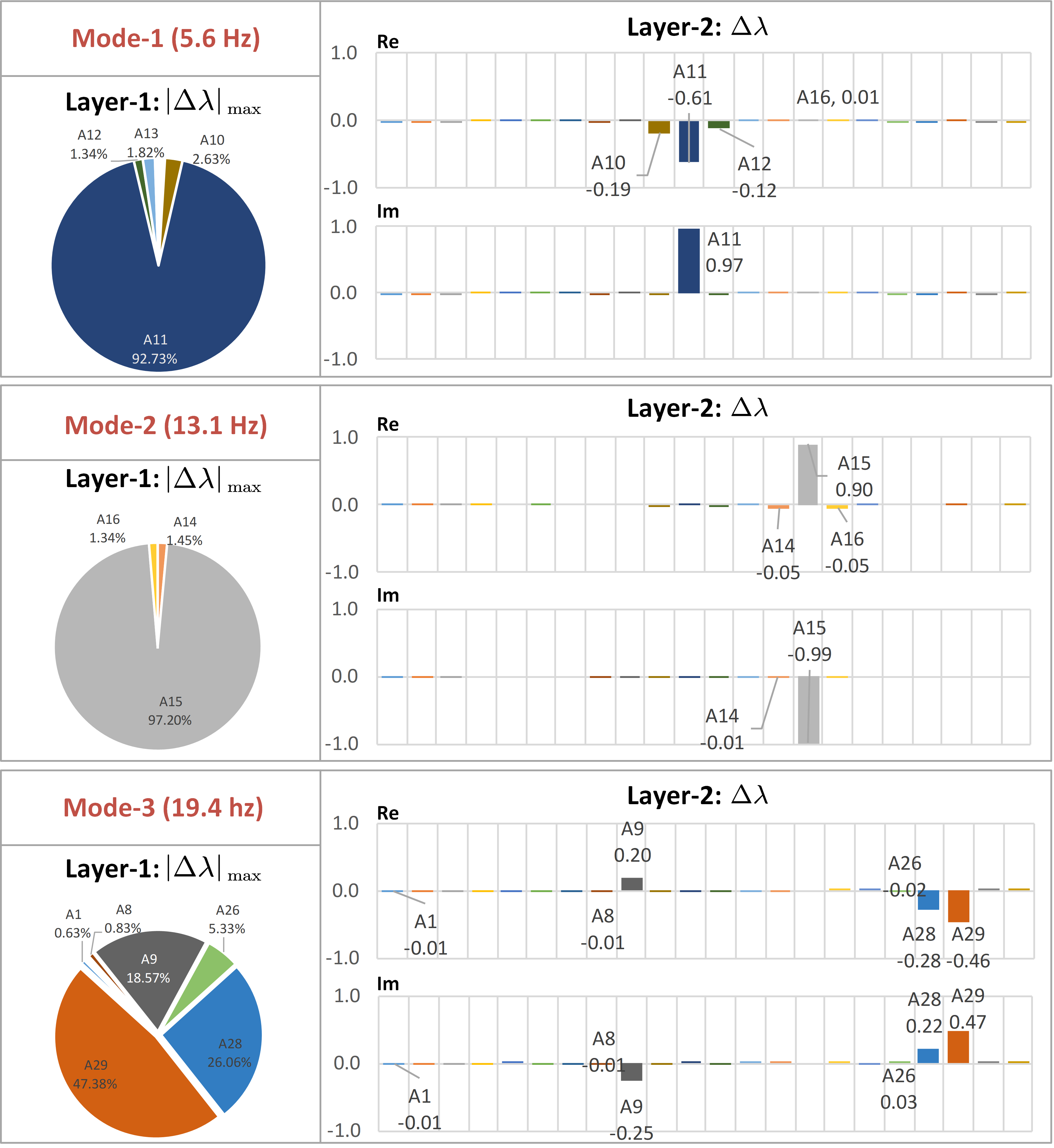}
	\caption{Participation analysis of the three under-damped modes using Layer-1 and Layer-2 of the grey-box, where the results in Layer-2 is normalized to the sum of absolute values.}
	\label{fig_level2}
	\vspace{-2mm}
\end{figure}

\begin{figure}
	\centering
	\includegraphics[width=3.0in]{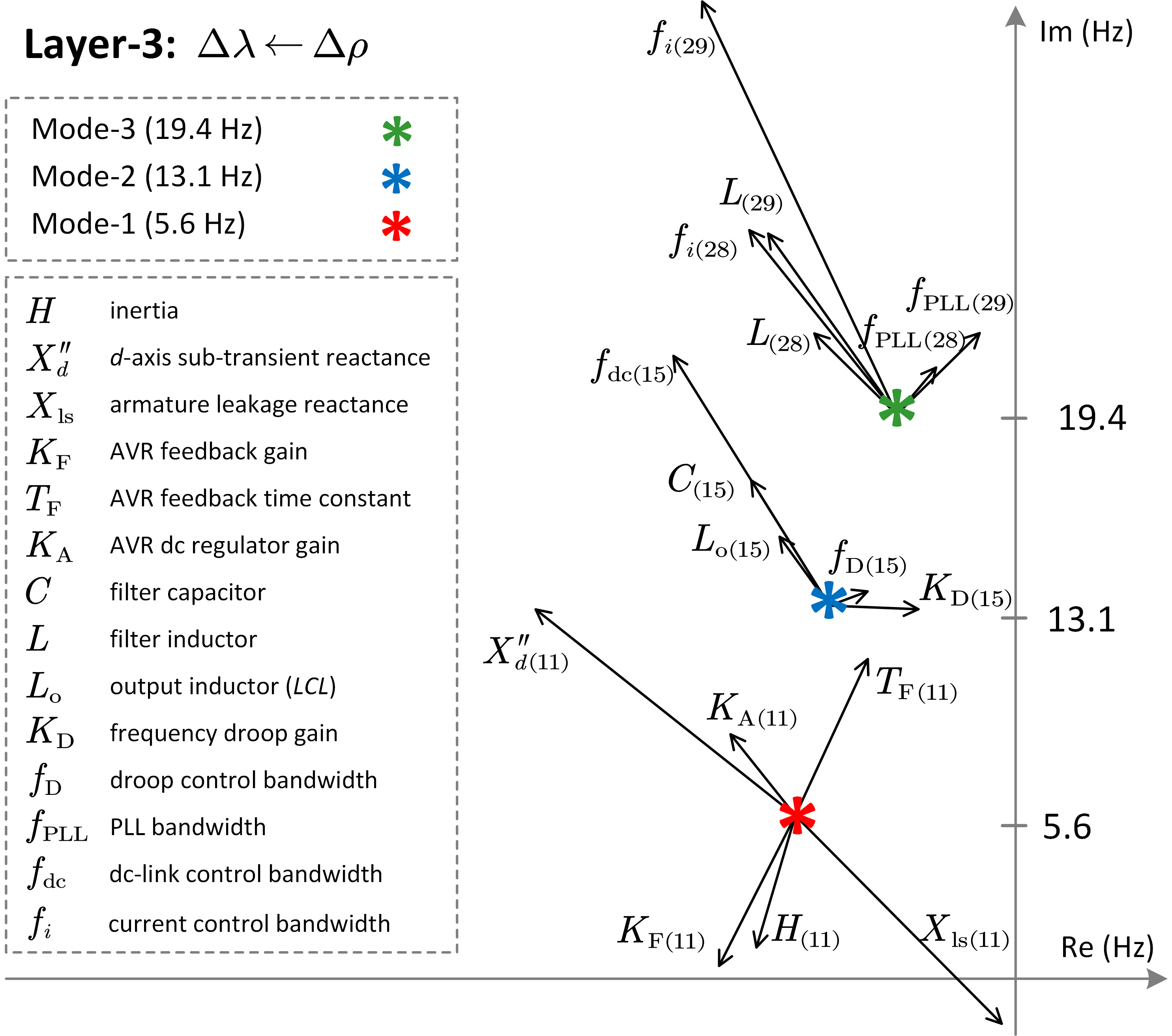}
	\caption{Grey-box Layer-3 analysis for the three under-damped modes. The parameter participation factors are represented as vectors around the associated modes showing the amplitude and direction of the eigenvalue variation subject to parameter perturbations. The parameter perturbations are proportional to the original value of parameters. The number in the parentheses in the subscript of each parameter indicates the associated apparatus number.}
	\label{fig_level3}
	\vspace{-4mm}
\end{figure}

After locating the participating apparatus and identifying their roles in system stability, the final step is to use the grey-box Layer-3 to re-tune the parameters in A11, A15, A28 and A29 to improve the stability of the three modes. Based on the chain-rule, the parameter participation factors for mode-1, mode-2 and mode-3 against the internal parameters of A11, A15, A28 and A29 were calculated and are illustrated in \figref{fig_level3}. It can be seen that mode-1 is sensitive to the physical parameters of A11 such as the sub-transient reactance, the armature leakage reactance and the inertia but also to control parameters in the automatic voltage regulator (AVR). In practice, it is easier to tune the control parameters, hence we can choose to increase the AVR feedback gain $K_{\text{F(11)}}$ to damp mode-1. For mode-2, both the dc-link control and the droop control have impacts on the damping but the mode is also sensitive to the $LCL$ filter capacitor and inductor. Decreasing the frequency droop gain $K_{\text{D(15)}}$ shifts the mode leftwards and stabilizes the system. This reflects the fact that $K_{\text{D(15)}}$ had been de-tuned high. For mode-3, parameters in A28 and A29 participate in a similar way with those in A29 having a larger impact. This reflects the fact that A29 had been de-tuned further than A28. Looking inside each inverter, it can be seen that increasing the current control bandwidth $f_{i}$ helps to stabilize the mode, but increasing the PLL bandwidth $f_{\text{PLL}}$ tends to destabilize the mode. We can remark, therefore, that this mode results from coupling between inner-loop (current control) and outer-loop (PLL) in a relatively weak grid. In this case, we choose to increase $f_{i(28)}$ and $f_{i(29)}$ to stabilize mode-3. For all the three modes, Layer-3 provides guidance on how to change control parameters to stabilize the system without the demand for changing the hardware. 

It is worth noting that many of the parameters are directly associated with states. For example, the parameter participation factor of $K_{\text{D}}$ is identical to the corresponding state participation factor of the droop control state. Thus Layer-3 provides similar interpretability to the classic state participation analysis. 

\begin{figure*}
\centering
\includegraphics[scale=0.3]{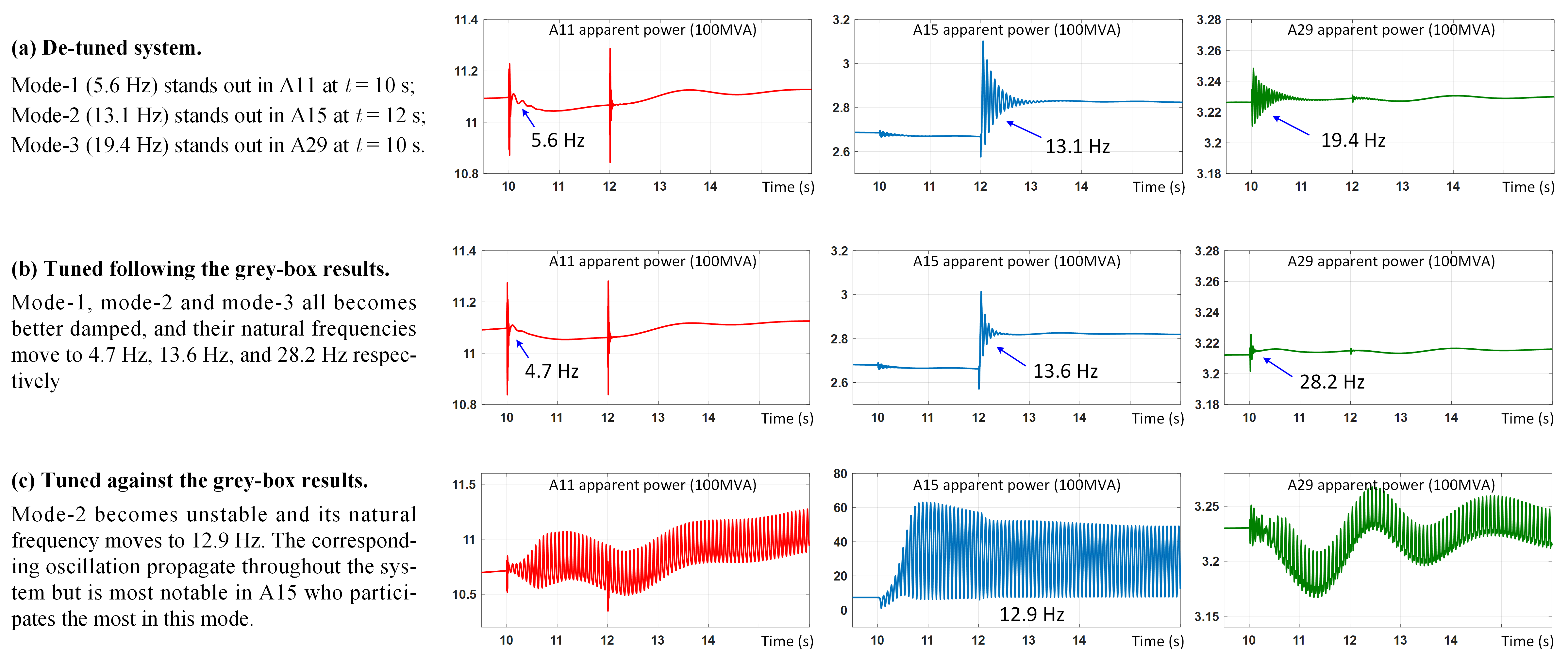}

\caption{Apparent power output of A11, A15 and A29 during two transients: load-61 disconnected at $t=10~\text{s}$, and load-42 increased 5\% at $t=12~\text{s}$. (a) De-tuned system with obvious oscillations during transient process; (b) Re-tuned by increase of $K_{\text{F(11)}}$ by 100\%, decrease of $K_{\text{D(15)}}$ by 60\%, and increase of  $f_{i(28)}$ and $f_{i(29)}$ each by 50\%, giving significant improvement in system stability; (c) Counter-tuned by increase of $K_{\text{D(15)}}$ by 50\% leading to instability.}
\label{fig_step}
\end{figure*}

To verify the predictions of the three layers of grey-box analysis, time-domain simulation of this case-study was conducted with the apparent power output of A11, A15 and A29 recorded in \figref{fig_step}. Two step-changes were introduced to the system to create transient behavior: the load at bus-61 was disconnected at $t=10~\text{s}$, and a $5\%$ increase in load at bus-42 was applied at $t=5~\text{s}$. Detailed description and discussion of the results are presented in the text within the figure and caption. Natural frequencies can be measured from the time-domain oscillations and they are found to agree with the resonant peaks in the frequency domain spectra in \figref{fig_admitt}. Different modes are excited in different apparatus, which agrees with the prediction of grey-box Layer-1 and Layer-2. The parameters of the participating apparatus are tuned with and against the suggestion of grey-box Layer-3 and the system is stabilized and destabilized accordingly. The grey-box-based participation analysis has correctly located the root-cause of oscillations and indicated appropriate choices for achieving stabilization.

\section{Conclusions}
The grey-box approach establishes a systematic method for small-signal stability and participation analysis of complex power systems with only impedance information. It has three layers with different transparencies at each layer to facilitate root-cause tracing to different depths, i.e. apparatus and parameters, according to the available knowledge. These grey-boxes provide a very useful tool to look inside a black-box (impedance) model to achieve almost the same transparency as a white-box (state-space) model, but without the need for manufacturers of apparatus to disclose the internal details that would be required in the white-box approach. The proposed grey-box approach is based on rigorous mathematical analysis with proof of the relationship between the residue and the impedance participation factor, and elucidation of the chain-rule of sensitivity propagation for internal states and parameters to be carried forward to the impedance participation factor, thus presenting a unified participation theory. 

\appendices

\section{Mathematical Preliminaries}
We summarise the mathematical preliminaries used in this paper to assist the reader and to make the paper self-contained.
\subsection{Residue}
In complex analysis, the residue of a complex function $G(s)$ is defined as the $g_{-1}$ coefficient of the Laurent series \cite{stein2010complex} of $G(s)$. This is, given the Laurent series of $G(s)$ around $\lambda$  
\begin{equation}
G(s) = \sum_{h=-\infty}^{\infty} g_h \cdot (s - \lambda)^{h}
\end{equation}
the residue of $G$ at $\lambda$ is defined as 
\begin{equation}
\text{Res}_\lambda G = g_{-1}.
\end{equation}
If $\lambda$ is a non-repeated pole of $G$, the residue is found from 
\begin{equation}
\text{Res}_\lambda G = \lim_{s\rightarrow \lambda} (s - \lambda) G(s).
\end{equation}
This property is used in the proof of Lemma 1 in Appendix B. The residue can be applied element-wise on a matrix of complex functions.

\subsection{Frobenius Inner Product} 
For two complex-valued matrices $V$ and $W$ with the same dimension, the Frobenius inner product \cite{garcia2017second} of $V$ and $W$ is defined as
\begin{equation}
\label{Frobenius_complex}
\langle V,W \rangle \triangleq \sum_{h,l}{\overline{{V}}_{hl} \  W_{hl}}
\end{equation}
where $h$ and $l$ are the row and column indices of the matrices, and $\overline{\phantom{x}}$ denotes complex conjugation. The complex conjugation in (\ref{Frobenius_complex}) ensures that the Frobenius inner product of a complex matrix with itself is a non-negative real number, and thus is induced the Frobenius norm $\| \cdot \|$
\begin{equation}
\label{Frobenius_norm}
\| V \| \triangleq \sqrt{\langle V,V \rangle}.
\end{equation}

The Frobenius inner product and norm are derived from the common inner product in vector spaces, so the properties of the common inner product are naturally inherited. One of the most useful properties is the Cauchy inequality
\begin{equation}
|\langle V,W \rangle| \leq \| V \| \cdot \| W \|
\end{equation}
where the equality holds if and only if $V$ and $W$ are aligned in orientation. This property is used in (\ref{cauchy}) in Section III-C. 

\subsection{Complex Function Derivative}
In this paper we use two types of derivatives for complex functions. The first type of derivative is the derivative of a complex function over a real number, e.g. derivative of a transfer function over its internal parameter $\rho$, $\partial H_\rho(s) / \partial{\rho}$. This type of derivative is the same as a real-function derivative with the real part and complex part of $H_\rho(s)$ treated separately. The second type of derivative is the derivative of a complex function over another complex number, e.g. a transfer function over another transfer function, $\partial H / \partial G$. For such a case, the mapping from $G$ to $H$ has to be analytic so that $\partial H / \partial G$ exists. For the scope of this paper, most complex-to-complex mappings are analytic throughout the complex plane except at poles so the derivative is almost always proper. As a special case, the derivative of a transfer function $H$ over its complex argument $s$ is called $H^\prime$, that is, 
\begin{equation}
H ^\prime(s) \triangleq \partial H(s) / \partial s.
\end{equation}

Both the complex and real derivatives can be applied to vectors and matrices with each element of the vectors and matrices treated as independent variables, and the resulted derivative is also a vector or matrix.

\section{Proof of Lemma 1}

We first prove the reduced case where $G_\rho$ is a scalar transfer function and the pole $\lambda$ is a zero of $H_\rho = G_\rho^{-1}$, that is,
\begin{equation}
\label{H_lambda}
    H_\rho(\lambda) = 0.
\end{equation}
A perturbation in $\rho$ induces a corresponding perturbation in $\lambda$, that is
\begin{equation}
\label{Delta_H_lambda_}
    H_{\rho+\Delta \rho}(\lambda+\Delta \lambda) = 0.
\end{equation}
Since $H_{\rho}$ is analytic around its zero $\lambda$, we have the following first-order Taylor expansion of (\ref{Delta_H_lambda_}) 
\begin{equation}
\label{Taylor}
    H_{\rho+\Delta \rho}(\lambda) + H_{\rho+\Delta \rho}^\prime(\lambda) \Delta \lambda= 0
\end{equation}
in which $H^\prime$ represents the derivative of $H$. Combining (\ref{H_lambda})-(\ref{Taylor}) yields
\begin{multline}
    H_{\rho+\Delta \rho}(\lambda) - H_{\rho}(\lambda) +  \\ H_{\rho}^\prime(\lambda) \Delta \lambda +
    { \left(H_{\rho+\Delta \rho}^\prime(\lambda) - H_{\rho}^\prime(\lambda)\right) \Delta \lambda} = 0
\end{multline}
and equivalently
\begin{equation}
\label{Delta_H_lambda}
    \Delta H_{\rho}(\lambda) + H_{\rho}^\prime(\lambda) \Delta \lambda + \Delta H_{\rho}^\prime(\lambda) \Delta \lambda= 0.
\end{equation}	
Suppressing the high-order infinitesimal $\Delta H_{\rho}^\prime(\lambda) \Delta \lambda$ in (\ref{Delta_H_lambda}) yields 
\begin{equation}
\label{suppress}
     \Delta \lambda = -H_{\rho}^\prime(\lambda)^{-1} 
     \cdot \Delta H_{\rho}(\lambda).
\end{equation}
As $\lambda$ is a non-repeated pole for $G_\rho$, the residue of $G_\rho$ at $\lambda$ is
\begin{equation}
\label{LHopital}
     \text{Res}_{\lambda} G_\rho = \lim_{s \rightarrow \lambda} (s - \lambda)G_\rho(s) =
     \lim_{s \rightarrow \lambda} \frac{s-\lambda}{H_\rho(s)} = \frac{1}{H_{\rho}^\prime(\lambda)}
\end{equation}
in which the second equal sign results from L'Hôpital's rule. Combining (\ref{suppress}) and (\ref{LHopital}) yields
\begin{equation}
     \Delta \lambda = 
     -\text{Res}_{\lambda} G_\rho \cdot 
     \Delta H_{\rho}(\lambda).
\end{equation}
This is the reduced case of Lemma 1 with $G_\rho$ being a scalar transfer function. 

Now we prove the case where $G_\rho$ is a square matrix and the pole $\lambda$ is a zero for the determinant of $H_\rho$, that is, $\text{det} (H_\rho(\lambda)) = 0$.
We take $\text{det} (H_\rho) \triangleq H_\text{det}$ as a scalar transfer function so a similar result to (\ref{suppress}) is obtained
\begin{equation}
\label{Delta_lambda_det}
     \Delta \lambda = -H_\text{det}^\prime(\lambda)^{-1} \Delta H_\text{det}(\lambda).
\end{equation}
Expanding $H_\text{det}$ along a column yields
\begin{equation}
\label{expand_by_element}
     H_\text{det} = \sum_{h} H_{\rho hl} {F}_{\rho hl}
\end{equation}
in which ${F}_{\rho}$ is the cofactor matrix for ${H}_{\rho}$ and the subscript $hl$ denotes the element in a matrix at the $h$-th row and $l$-th column. It is clear to see from (\ref{expand_by_element}) that \begin{equation}
     \frac{\partial H_\text{det}}{\partial H_{\rho hl}} = {F}_{\rho hl}
\end{equation} 
and hence
\begin{equation}
\label{Delta_H_det}
\begin{split}
    \Delta H_\text{det} = & 
    \sum_{h,l} \frac{\partial H_\text{det}}{\partial H_{\rho hl}} \Delta H_{\rho hl} \\ 
    = &
    \sum_{h,l} {F}_{\rho hl} \Delta H_{\rho hl} = 
    \langle \overline{{F}}_{\rho}, \Delta H_{\rho} \rangle
\end{split}
\end{equation}
where the complex conjugate $\overline{\phantom{x}}$ is associated with the Frobenius inner product $\langle \cdot , \cdot \rangle$ for complex matrices defined in (\ref{Frobenius_complex}). 

Since $G_\rho$ is now a matrix, its residue needs to be calculated element-wise
\begin{equation}
\begin{split}
\label{residue_by_element}
     \text{Res}_\lambda G_\rho 
     = &
     \lim_{s \rightarrow \lambda} 
     {(s-\lambda)}{G_\rho(s)} \\
     = &
     \lim_{s \rightarrow \lambda} \left({(s-\lambda)}{H_\rho(s)^{-1}}\right) \\
     = & 
     \lim_{s \rightarrow \lambda} \left(\frac{s-\lambda}{H_{\text{det}}(s)}{{F}_\rho(s)^{\top}}\right) =
     \frac{{F}_\rho(\lambda)^{\top}}{H^{\prime}_{\text{det}}(\lambda)}
\end{split}
\end{equation}
where we make use of the fact that
\begin{equation}
    H_\rho(s)^{-1} = F_\rho(s)^{\top} / H_{\text{det}(s)}.
\end{equation}
Combining (\ref{Delta_lambda_det}), (\ref{Delta_H_det}) and (\ref{residue_by_element}) yields Lemma 1
\begin{equation}
\begin{split}
	&\Delta \lambda = -H_{\text{det}}^\prime(\lambda)^{-1} 
    \langle \overline{F_{\rho}(\lambda)}, \Delta H_{\rho}(\lambda) \rangle \\
    &= \langle -\overline{\text{Res}_\lambda G_\rho}^\top , \Delta H_\rho (\lambda) \rangle 
    = \langle -\text{Res}^*_\lambda G_\rho , \Delta H_\rho (\lambda) \rangle .
\end{split}    
\end{equation}

\section{Illustration of LEMMA 1}

We use a simple three-node system to illustrate Lemma 1. The whole-system admittance of this three-node system is
	\begin{equation}
		{\hat{Y}}=\left[ \begin{matrix}
			\hat{Y}_{11}&		\hat{Y}_{12}&		\hat{Y}_{13}\\
			\hat{Y}_{21}&		\hat{Y}_{22}&		\hat{Y}_{23}\\
			\hat{Y}_{31}&		\hat{Y}_{32}&		\hat{Y}_{33}\\
		\end{matrix} \right] 
	\end{equation}
where each entry of ${\hat{Y}}$ is a $2\times2$ transfer function matrix in $dq$ frame. For instance, the first diagonal element is 
\begin{equation}
	\hat{Y}_{11}=\left[ \begin{matrix}
		\hat{Y}_{11}^{dd}\left( s \right)&	\hat{Y}_{11}^{dq}\left( s \right)\\
		\hat{Y}_{11}^{qd}\left( s \right)&	\hat{Y}_{11}^{qq}\left( s \right)\\
	\end{matrix} \right]
\end{equation}
which represents the whole-system admittance measured at the first node. 
Each element in $\hat{Y}_{11}$ can be expressed as the sum of a series of pole-residue pairs, for example,  $\hat{Y}_{11}^{dd}$ is
\begin{equation}
	\hat{Y}_{11}^{dd}\left( s \right) =\frac{r_{11,1}^{dd}}{s-\lambda _1}+\frac{r_{11,2}^{dd}}{s-\lambda _2}+\cdots +\frac{r_{11,N}^{dd}}{s-\lambda _N}
\end{equation}
where $r_{11,n}^{dd}$ is the residue of $\hat{Y}_{11}^{dd}$ at the $n$-th pole (eigenvalue) $\lambda_n$ for $n \in \{1,2,\cdots,N \}$. The residue of $\hat{Y}_{11}$ at a particular pole $\lambda$ (subscript $n$ is dropped for brevity) is then given by
\begin{equation}
	\mathrm{Res}_{\lambda}\hat{Y}_{11}=\left[ \begin{matrix}
		r_{11}^{dd}&		r_{11}^{dq}\\
		r_{11}^{qd}&		r_{11}^{qq}\\
	\end{matrix} \right]
\end{equation}
which yields the impedance participation factor of the apparatus connected at the first node according to (\ref{impedance_participation})
\begin{equation}
	p_{\lambda ,Z_1}=-\mathrm{Res}_{\lambda}^{*}\hat{Y}_{11}=-\left[ \begin{matrix}
		\overline{r}_{11}^{dd}&		\overline{r}_{11}^{qd}\\
		\overline{r}_{11}^{dq}&		\overline{r}_{11}^{qq}\\
	\end{matrix} \right]. 
\end{equation}

For a parameter perturbation $\Delta\rho$ in the apparatus at the first node, the corresponding impedance perturbation is
\begin{equation}
	\Delta Z_1\left( \lambda \right) =  
	\frac{\partial Z_1\left( \lambda \right)}{\partial \rho}
	\cdot
	\Delta \rho
	=\left[ \begin{matrix}
		\Delta Z_{1}^{dd}&		\Delta Z_{1}^{dq}\\
		\Delta Z_{1}^{qd}&		\Delta Z_{1}^{qq}\\
	\end{matrix} \right]
\end{equation}
which yields 
\begin{equation}
\begin{split}
	& \Delta \lambda =\left< p_{\lambda ,Z_1},\Delta Z_1\left( \lambda \right) \right> 
	\\
	& =-( r_{11}^{dd}\Delta Z_{2}^{dd} + r_{11}^{qd}\Delta Z_{2}^{dq}  + r_{11}^{dq}\Delta Z_{2}^{qd} +  r_{11}^{qq}\Delta Z_{2}^{qq})
\end{split}
\end{equation}
according to Lemma 1. $\Delta \lambda$ is a complex number whose direction is determined jointly by $p_{\lambda, Z_1}$ and $\Delta Z_1\left( \lambda \right)$.

\color{black}

\ifCLASSOPTIONcaptionsoff
  \newpage
\fi

\bibliographystyle{IEEEtran}
\bibliography{References}

\begin{thebibliography}{10}
\providecommand{\url}[1]{#1}
\csname url@samestyle\endcsname
\providecommand{\newblock}{\relax}
\providecommand{\bibinfo}[2]{#2}
\providecommand{\BIBentrySTDinterwordspacing}{\spaceskip=0pt\relax}
\providecommand{\BIBentryALTinterwordstretchfactor}{4}
\providecommand{\BIBentryALTinterwordspacing}{\spaceskip=\fontdimen2\font plus
\BIBentryALTinterwordstretchfactor\fontdimen3\font minus
  \fontdimen4\font\relax}
\providecommand{\BIBforeignlanguage}[2]{{%
\expandafter\ifx\csname l@#1\endcsname\relax
\typeout{** WARNING: IEEEtran.bst: No hyphenation pattern has been}%
\typeout{** loaded for the language `#1'. Using the pattern for}%
\typeout{** the default language instead.}%
\else
\language=\csname l@#1\endcsname
\fi
#2}}
\providecommand{\BIBdecl}{\relax}
\BIBdecl

\bibitem{bialek2020does}
J.~Bialek, ``What does the {GB} power outage on 9 august 2019 tell us about the
  current state of decarbonised power systems?'' \emph{Energy Policy}, vol.
  146, p. 111821, 2020.

\bibitem{gu2019motion}
Y.~Gu, J.~Liu, T.~C. Green, W.~Li, and X.~He, ``Motion-induction compensation
  to mitigate sub-synchronous oscillation in wind farms,'' \emph{IEEE
  Transactions on Sustainable Energy}, vol.~11, no.~3, pp. 1247--1256, 2019.

\bibitem{li2021impedance}
Y.~Li, Y.~Gu, Y.~Zhu, A.~Junyent-Ferr{\'e}, X.~Xiang, and T.~C. Green,
  ``Impedance circuit model of grid-forming inverter: Visualizing control
  algorithms as circuit elements,'' \emph{IEEE Transactions on Power
  Electronics}, vol.~36, no.~3, pp. 3377--3395, 2021.

\bibitem{wang2017unified}
X.~Wang, L.~Harnefors, and F.~Blaabjerg, ``Unified impedance model of
  grid-connected voltage-source converters,'' \emph{IEEE Transactions on Power
  Electronics}, vol.~33, no.~2, pp. 1775--1787, 2017.

\bibitem{rommes2008computing}
J.~{Rommes} and N.~{Martins}, ``Computing large-scale system eigenvalues most
  sensitive to parameter changes, with applications to power system
  small-signal stability,'' \emph{IEEE Transactions on Power Systems}, vol.~23,
  no.~2, pp. 434--442, May 2008.

\bibitem{Sinha2019}
S.~{Sinha}, P.~{Sharma}, U.~{Vaidya}, and V.~{Ajjarapu}, ``On information
  transfer-based characterization of power system stability,'' \emph{IEEE
  Transactions on Power Systems}, vol.~34, no.~5, pp. 3804--3812, 2019.

\bibitem{gu2020impedance}
Y.~Gu, Y.~Li, Y.~Zhu, and T.~Green, ``Impedance-based whole-system modeling for
  a composite grid via embedding of frame dynamics,'' \emph{IEEE Transactions
  on Power Systems}, 2020.

\bibitem{Huang2007}
Z.~{Huang}, Y.~{Cui}, and W.~{Xu}, ``Application of modal sensitivity for power
  system harmonic resonance analysis,'' \emph{IEEE Transactions on Power
  Systems}, vol.~22, no.~1, pp. 222--231, 2007.

\bibitem{ebrahimzadeh2018bus}
E.~{Ebrahimzadeh}, F.~{Blaabjerg}, X.~{Wang}, and C.~L. {Bak}, ``Bus
  participation factor analysis for harmonic instability in power electronics
  based power systems,'' \emph{IEEE Transactions on Power Electronics},
  vol.~33, no.~12, pp. 10\,341--10\,351, Dec 2018.

\bibitem{zhan2019frequency}
Y.~Zhan, X.~Xie, H.~Liu, H.~Liu, and Y.~Li, ``Frequency-domain modal analysis
  of the oscillatory stability of power systems with high-penetration
  renewables,'' \emph{IEEE Transactions on Sustainable Energy}, vol.~10, no.~3,
  pp. 1534--1543, 2019.

\bibitem{kundur1994power}
P.~Kundur, \emph{Power system stability and control}.\hskip 1em plus 0.5em
  minus 0.4em\relax McGraw-hill New York, 1994, vol.~7.

\bibitem{Perezarriaga1982}
I.~J. {Perez-arriaga}, G.~C. {Verghese}, and F.~C. {Schweppe}, ``Selective
  modal analysis with applications to electric power systems, part i: Heuristic
  introduction,'' \emph{IEEE Transactions on Power Apparatus and Systems}, vol.
  PAS-101, no.~9, pp. 3117--3125, 1982.

\bibitem{Yang2019Automation}
D.~Yang, X.~Wang, M.~Ndreco, W.~Winter, R.~Juhlin, and A.~Krontiris,
  ``\BIBforeignlanguage{English}{Automation of impedance measurement for
  harmonic stability assessment of mmc-hvdc systems},'' in
  \emph{\BIBforeignlanguage{English}{18th International Workshop on Large-scale
  Integration of Wind Power into Power Systems}}, 2019.

\bibitem{LinearizationMatlab}
\BIBentryALTinterwordspacing
{Mathworks Help Center}, ``Exact linearization algorithm.'' [Online].
  Available:
  \url{https://uk.mathworks.com/help/slcontrol/ug/exact-linearization-algorithm.html}
\BIBentrySTDinterwordspacing

\bibitem{Simplex}
\BIBentryALTinterwordspacing
``Future power networks.'' [Online]. Available:
  \url{https://github.com/Future-Power-Networks/Publications}
\BIBentrySTDinterwordspacing

\bibitem{Roinila2013Broadband}
T.~{Roinila}, M.~{Vilkko}, and J.~{Sun}, ``Broadband methods for online grid
  impedance measurement,'' in \emph{2013 IEEE Energy Conversion Congress and
  Exposition}, 2013, pp. 3003--3010.

\bibitem{Familiant2009New}
Y.~A. {Familiant}, J.~{Huang}, K.~A. {Corzine}, and M.~{Belkhayat}, ``New
  techniques for measuring impedance characteristics of three-phase ac power
  systems,'' \emph{IEEE Transactions on Power Electronics}, vol.~24, no.~7, pp.
  1802--1810, 2009.

\bibitem{Gustavsen1999}
B.~{Gustavsen} and A.~{Semlyen}, ``Rational approximation of frequency domain
  responses by vector fitting,'' \emph{IEEE Transactions on Power Delivery},
  vol.~14, no.~3, pp. 1052--1061, 1999.

\bibitem{boyd2018introduction}
S.~Boyd and L.~Vandenberghe, \emph{Introduction to applied linear algebra:
  vectors, matrices, and least squares}.\hskip 1em plus 0.5em minus 0.4em\relax
  Cambridge university press, 2018.

\bibitem{Pal_Report2013}
A.~K. Singh and B.~Pal, ``Ieee pes task force on benchmark systems for
  stability controls: Report on the 68-bus, 16-machine, 5-area system,'' 12
  2013.

\bibitem{Qi_Identification}
B.~{Qi}, K.~N. {Hasan}, and J.~V. {Milanović}, ``Identification of critical
  parameters affecting voltage and angular stability considering load-renewable
  generation correlations,'' \emph{IEEE Transactions on Power Systems},
  vol.~34, no.~4, pp. 2859--2869, 2019.

\bibitem{stein2010complex}
E.~M. Stein and R.~Shakarchi, \emph{Complex analysis}.\hskip 1em plus 0.5em
  minus 0.4em\relax Princeton University Press, 2010, vol.~2.

\bibitem{garcia2017second}
S.~R. Garcia and R.~A. Horn, \emph{A second course in linear algebra}.\hskip
  1em plus 0.5em minus 0.4em\relax Cambridge University Press, 2017.

\end{thebibliography}

\end{document}